\documentclass[aip,jmp,12pt]{revtex4-1}

\usepackage[T1]{fontenc}
\usepackage[latin1]{inputenc}
\usepackage{amsmath,amsthm,amssymb}
\usepackage{dsfont}

\newcommand{\ud}{\mathrm{d}}

\newcommand{\cD}{{\mathcal D}}

\newcommand{\rz}{{\mathbb R}}
\newcommand{\nz}{{\mathbb N}}

\DeclareMathOperator{\supp}{supp}

\newcommand{\eins}{\boldsymbol{1}}

\newtheorem{theorem}{Theorem}[section]
\newtheorem{lemma}[theorem]{Lemma}

\newtheorem{cor}[theorem]{Corollary}

\theoremstyle{definition}

\newtheorem{rem}[theorem]{Remark}

\begin{document}

\title{Two interacting particles on the half-line}

\author{Joachim Kerner }
\email[]{joachim.kerner@fernuni-hagen.de}
\affiliation{Department of Mathematics and Computer Science, FernUniversit\"{a}t in Hagen, 58084 Hagen, Germany}

\author{Tobias M\"{u}hlenbruch}
\email[]{tobias.muehlenbruch@fernuni-hagen.de}
\affiliation{Department of Mathematics and Computer Science, FernUniversit\"{a}t in Hagen, 58084 Hagen, Germany}

\date{\today}

\begin{abstract} In the case of general compact quantum graphs, many-particle models with singular two-particle interactions where introduced in Refs.~\citenum{BKSingular,BKContact} in order to provide a paradigm for further studies on many-particle quantum chaos. In this note, we discuss various aspects of such singular interactions in a two-particle system restricted to the half-line $\rz_+$. Among others, we give a description of the spectrum of the two-particle Hamiltonian and obtain upper bounds on the number of eigenstates below the essential spectrum. We also specify conditions under which there is exactly one such eigenstate. As a final result, it is shown that the ground state is unique and decays exponentially as $\sqrt{x^2+y^2} \to \infty$.
\end{abstract}



\maketitle 

\section{Introduction}
Quantum graphs have proven to be useful models in various areas of mathematical and theoretical physics. From a mathematical point of view, they combine the simplicity of a (quasi) one-dimensional system with the complexity of a graph-like structure. It is exactly this underlying complexity that turned quantum graphs into particular important models in the field of quantum chaos by showing that eigenvalue correlations in quantum graphs are generically the same as in quantum systems with chaotic classical limit~\cite{KS97,GNUSMY06}. Intuitively, the chaotic behavior stems from the scattering of the quantum particle in the vertices of the graph. Therefore, in order to provide a useful model for the investigation of many-particle quantum chaos, many-particle systems on (finite, compact) quantum graphs with localised many-particle interactions were introduced in Refs.~\citenum{BKSingular} and \citenum{BKContact}. Localised in this context means that the particles interact only in the vicinity of the vertices of the graph, i.e., one particle has to sit at a vertex whereas the other particles are close to it. The important consequence of such singular many-particle interactions is that the scattering of the particles and hence the dynamics of the system are affected in a way encoding genuine many-particle correlations.

Also, apart from the fact that the model which will be discussed in this paper originated from the field of quantum chaos, it is also an interesting model in its own right. From the point of view of applications, singular many-particle interactions on graphs were already discussed in Ref.~\citenum{MP95} in order to investigate the effects of short-range many-body interactions on, e.g., the conductivity of nanoelectronic devices. More explicitly, our model can be understood as a toy-model to investigate a system of two electrons moving in a wire which is normal-conducting except for a relatively small part where it is super-conducting~\cite{FossheimSuperconducting,MP95}. In the super-conducting part the electrons are then interacting due to the pairing effect in superconductivity (Cooper pairs).

From a more theoretical point of view, it is well-known that quantum many-body problems are generally hard to solve and, indeed, there are only few models which are explicitly solvable \cite{AGHH88}. One important model in this respect is the Lieb-Liniger model (originally formulated for an $N$-particle system), see Ref.~\citenum{LL63}, whose Hamiltonian, in the case of two particles, is formally given by
\begin{equation}\label{LiebLinigerHamiltonian}
H=-\frac{\partial^2}{\partial x^2}-\frac{\partial^2}{\partial y^2}+\alpha \delta(x-y)\ ,
\end{equation}
where $\alpha \in \rz_+$ is the interaction strength and $\delta(x)$ the Dirac delta function. Explicitly solvable for this model means that its eigenfunctions can be (if it is considered on an interval together with periodic boundary conditions) explicitly given and the spectrum can, at least in principle, be calculated. We note that versions of the Lieb-Liniger model on graphs were discussed in Refs.~\citenum{Ha07} and \citenum{Ha08} in the case of star graphs and in Ref.~\citenum{BKContact} in the case of general compact graphs. We will see below that the formal Hamiltonian of our model resembles the Hamiltonian \eqref{LiebLinigerHamiltonian} to some extent. The difference is that the two-particle interactions which we will consider are even more singular in the sense that they do not solely depend on the relative position of the two particles. It is important to realise that, although the Lieb-Liniger model was originally considered as being of theoretical interest only, it is nowadays recognised to describe realistic gases in one dimension very well at low temperatures~\cite{cazalilla2011one}. Hence, in the same spirit as for the Lieb-Liniger model, we consider the model to be discussed as an interesting model for theoretical as well as for practical purposes.

This paper is organized as follows: In Section~\ref{Sec1} we introduce our model and give a precise mathematical formulation, i.e., we establish a quadratic form which is then shown to characterise a self-adjoint operator (the Hamiltonian of our system) uniquely. Section~\ref{Sec2} is devoted to the spectral analysis of this operator, i.e., we will characterise its essential spectrum and establish conditions for which there exists at least one eigenstate below the essential spectrum. We will also provide upper bounds on the number of such eigenstates. Finally, in Section~\ref{Sec3}, we will establish (pointwise) upper bounds on the ground state eigenfunction proving an exponential decay as $\sqrt{x^2+y^2} \to \infty$.

\section{The model and its Hamiltonian}
\label{Sec1}
The model we would like to discuss in this paper consists of two (distinguishable) interacting particles where the one-particle configuration space is given by the half-line $\rz_+=(0,\infty)$. The two-particle interactions shall be of singular type, i.e., the formal Hamiltonian of the system reads
\begin{equation}\label{FormalHamiltonian}
H=-\frac{\partial^2}{\partial x^2}-\frac{\partial^2}{\partial y^2}+v(x,y)\bigl[\delta(x)+\delta(y)\bigr]\ ,
\end{equation}
where $v(x,y)=v(y,x)$ is some symmetric interaction potential. On a physical level we can read off from \eqref{FormalHamiltonian} that the two particles interact only when at least one particle is situated at the origin. Furthermore, choosing $v(x,y)$ such that $\supp v \subset B_\varepsilon(0) \cap \rz^2_+$, where $B_\varepsilon(0) \subset \rz^2$ is the open ball with radius $\varepsilon > 0$ around $0 \in \rz^2$, we see that the particles interact only whenever one particle is situated at the origin and the other particle is contained in a small neighborhood around it.

On a mathematical level, the methods employed in Ref.~\citenum{BKSingular} show how the Hamiltonian \eqref{FormalHamiltonian} can be realised via a quadratic form on $L^2(\rz_+^2)$. Most interestingly, this quadratic form then corresponds to a variational formulation of a boundary-value problem for the two-dimensional Laplacian
\begin{equation}
-\Delta=-\frac{\partial^2}{\partial x^2}-\frac{\partial^2}{\partial y^2}\
\end{equation}
with coordinate dependent Robin boundary conditions. Indeed, defining $\sigma(y):=-v(0,y)$, the boundary conditions read
\begin{equation}\label{RBConditions}\begin{split}
\frac{\partial \varphi}{\partial n}(0,y)+\sigma(y)\varphi(0,y)&=0\ , \quad \text{and} \\
\frac{\partial \varphi}{\partial n}(y,0)+\sigma(y)\varphi(y,0)&=0\ ,
\end{split}
\end{equation}
where $\frac{\partial}{\partial n}$ denotes the inward normal derivative along $\partial \rz^2_+$.
\begin{rem}
We require $\sigma(y)$ to be a measurable, essentially bounded function throughout this paper.
\end{rem}
The associated quadratic form is given by
\begin{equation}\label{DefinitionQF}
q[\varphi]=\int_{\rz_+^2} |\nabla \varphi|^2\ \ud x - \int_{\partial \rz_+^2}\  \sigma(y)\ |\varphi_{bv}(y)|^2 \ \ud y\ ,
\end{equation}
defined on $\cD_q=H^1(\rz_+^2)$, i.e., the Sobolev space of order one. The $\varphi_{bv}(y)$ are the boundary values of $\varphi \in H^1(\rz_+^2)$ which are, according to the trace theorem for Sobolev functions (see Ref.~\citenum{Dob05}), well defined. It is then readily verified that the form \eqref{DefinitionQF} is meaningful and we can establish the following result.
\begin{theorem}\label{QuadraticFormClosedness} The form \eqref{DefinitionQF} is closed and semi-bounded.
\end{theorem}
\begin{proof} For the proof we note that the restrictions of all functions in $C^{\infty}_0(\rz^2)$ onto $\rz^2_+$ are dense in $H^1(\rz^2_+)$ (see, e.g., Ref.~\citenum{Dob05}). This allows us to transfer trace estimates for bounded (Lipschitz) domains to the unbounded case, i.e., to $\rz^2_+$. More explicitly, in the proof of Theorem~3.1 in Ref.~\citenum{BKSingular}, an estimate is used which extends to our case: we have
\begin{equation}
\|\varphi_{bv}\|^2_{L^2(\partial \rz^2_+)} \leq 4\left(\frac{2}{\delta}\|\varphi\|^2_{L^2(\rz^2_+)}+\delta\|\nabla \varphi\|^2_{L^2(\rz^2_+)}\right)
\end{equation}

for all $\delta > 0$. We hence obtain
\begin{equation}\label{EstimatesEquiv}
\left| \int_{\partial \rz_+^2}\  \sigma(y)\ |\varphi_{bv}(y)|^2 \ \ud y \right| \leq 4\ \|\sigma(y)\|_{\infty}\ \left(\frac{2}{\delta}\|\varphi\|^2_{L^2(\rz^2_+)}+\delta\|\nabla \varphi\|^2_{L^2(\rz^2_+)}\right)\ ,
\end{equation}
which implies
\begin{equation}\label{EquationLowerBoundSpectrum}
q[\varphi] \geq \bigl(1-4\delta \|\sigma(y)\|_{\infty}\bigr)\ \|\nabla \varphi\|^2_{L^2(\rz^2_+)}-\frac{8\|\sigma(y)\|_{\infty}}{\delta}\|\varphi\|^2_{L^2(\rz^2_+)}\ .
\end{equation}
Hence, choosing $\delta$ small enough, we see that the form $q[\cdot]$ is bounded from below.

Furthermore, the estimate \eqref{EstimatesEquiv} also implies that the form-norm is equivalent to the $H^1$-norm. Since $H^1(\rz^2_+)$ is complete, we conclude that $q[\cdot]$ is indeed closed.
\end{proof}
\begin{rem}\label{RemarkLowerBound}
From equation~\eqref{EquationLowerBoundSpectrum} one directly concludes that 
\begin{equation}
\inf \sigma(-\Delta_{\sigma}) \geq -32\|\sigma(y)\|_{\infty}^2
\end{equation}
by setting $\delta=\frac{1}{4\|\sigma(y)\|_{\infty}}$. Note that a sharper bound will be derived later in Theorem~\ref{EstimateGroundStateEnergy}.
\end{rem}
Due to the representation theorem of quadratic forms (see, e.g., Ref.~\citenum{BEH08}), Theorem~\ref{QuadraticFormClosedness} implies the existence of a unique self-adjoint operator being associated to the form $q[\cdot]$. This operator is the Hamiltonian of our system which we denote by $-\Delta_{\sigma}$.
\begin{rem} We note that the case $\sigma \equiv 0$ corresponds to the so called Neumann-Laplacian on $\rz^2_+$ being self-adjoint on the domain $\cD_{N}:=\{\varphi \in H^2(\rz^2_+)\ |\ \frac{\partial \varphi}{\partial n}=0 \ \text{on}\ \partial \rz^2_+\}$.
\end{rem}
\section{Spectral properties of $-\Delta_{\sigma}$}
\label{Sec2}
In this section we characterise the spectrum of $-\Delta_{\sigma}$, i.e., we describe the essential part as well as the discrete part of the spectrum. We will establish conditions on $\sigma(y)$ for which there exists at least one eigenstate below the essential spectrum. Furthermore, we provide upper bounds on the number of eigenstates below the essential spectrum and specify conditions on $\sigma(y)$ for which there is exactly one such eigenstate.
\subsection{On the essential spectrum}
In our first result we prove that the interval $[0,\infty)$ is, in the most general case of (essentially) bounded $\sigma(y)$, always contained in the essential spectrum. Furthermore, we prove that the essential spectrum has no negative part if either $\sigma(y)$ becomes eventually non-positive or if $|\sigma(y)|$ converges to zero. 
\begin{theorem}\label{TheoremNeumann} For $\sigma(y) \in L^{\infty}(\rz_+)$ one has
\begin{equation}\label{EquationContainingEssentialSpectrum}
[0,\infty) \subset \sigma_{ess}(-\Delta_{\sigma})\ .
\end{equation}
Furthermore, if either
\begin{enumerate}
	\item[(i)] there exists $L >0$ such that $\sigma(y) \leq 0$ for a.e. $y \geq L$\ ,
\end{enumerate}
or
\begin{enumerate}
\item[(ii)] $\mathrm{ess}\lim_{y \rightarrow \infty}|\sigma(y)|=0$\ , \quad (essential limit)
\end{enumerate}
then
\begin{equation}
\sigma_{ess}(-\Delta_{\sigma})=[0,\infty)\ .
\end{equation}
\end{theorem}
\begin{proof} To prove \eqref{EquationContainingEssentialSpectrum} we employ Weyl's characterization of the essential spectrum in the sense of quadratic forms~\cite{stollmann2001caught}. This means that, in order to prove that $\lambda \geq 0$ is in the essential spectrum, we construct a sequence $(\varphi_n)_{n \in \nz} \subset H^1(\rz^2_+)$ such that $\|\varphi_n\|=1$, $\varphi_n \rightharpoonup 0$ (weak convergence) and
\begin{equation}
\sup\left\{\bigl|s(\varphi_n,u)-\lambda \langle \varphi_n, u\rangle_{L^2(\rz^2_+)}\bigr|\ \Bigl| \ u \in H^1(\rz^2_+), \|u\|_{H^1(\rz^2_+)} \leq 1 \right \} \longrightarrow 0\ ,
\end{equation}
as $n \to \infty$ and where $s(\cdot,\cdot)$ is the sesquilinear form associated with \eqref{DefinitionQF}. To construct such a sequence we consider the rectangle $R_n:=[0,L_n] \times [y_n,y_n+B_n]$ with some constants $L_n,B_n,y_n > 0$ and choose $\varphi_n$ to be the normalized ground state eigenfunction of the two-dimensional Laplacian (with Dirichlet boundary conditions) on $R_n$, extended by zero to all of $\rz^2_+$. A direct calculation then gives
\begin{equation}
s(\varphi_n,u)=\left(\frac{\pi^2}{L^2_n}+\frac{\pi^2}{B^2_n}\right)\langle \varphi_n,u\rangle_{L^2(\rz^2_+)}\ .
\end{equation}
Hence, letting $L_n \to \infty$, $\frac{\pi^2}{B^2_n} \to \lambda$ and $y_n \to \infty$ fast enough, we see that $(\varphi_n)_{n \in \nz}$ is an appropriate sequence showing that $\lambda \in \sigma_{ess}(-\Delta_{\sigma})$. Hence $[0,\infty) \subset \sigma_{ess}(-\Delta_{\sigma})$.

To prove the second statement we use a bracketing argument. For this, we split $\rz^2_+$ into two disjoint subsets: we write $\rz^2_+=D_{R} \ \dot{\cup}\ \Omega_{-}$ where $\Omega_{-}:=\rz^2_+ \setminus D_{R}$ and $D_R=(0,R) \times (0,R)$ with $R > 0$ some constant. As a comparison operator we consider the direct sum $-\Delta^{N}_{D_{R}} \oplus -\Delta^{N}_{\Omega_{-}}$ of Laplacians, the index $N$ referring to (additional) Neumann boundary conditions along the inner two line segments of $D_R$. Note that both operators, $-\Delta^{N}_{D_{R}}$ as well as $-\Delta^{N}_{\Omega_{-}}$, can be defined via their associated quadratic forms which are similar to \eqref{DefinitionQF}. More precisely, $-\Delta^{N}_{D_{R}}$ is the unique self-adjoint operator being associated with
\begin{equation}
q_{1}[\varphi]:=\int_{D_R} |\nabla \varphi|^2\ \ud x - \int_{\partial \rz^2_+ \cap B_{R}(0)}\  \sigma(y)\ |\varphi_{bv}(y)|^2 \ \ud y\ ,
\end{equation}
defined on $H^1(D_R)$. Furthermore, $-\Delta^{N}_{\Omega_{-}}$ is the unique self-adjoint operator being associated with
\begin{equation}
q_{2}[\varphi]:=\int_{\Omega_{-}} |\nabla \varphi|^2\ \ud x - \int_{\partial \rz^2_+ \setminus B_{R}(0)}\  \sigma(y)\ |\varphi_{bv}(y)|^2 \ \ud y\ ,
\end{equation}
defined on $H^1(\Omega_{-})$.

A standard bracketing argument of operators, see eg.\ Ref.~\citenum{BEH08}, then implies
\begin{equation}
\inf \sigma_{ess}\bigl(-\Delta^{N}_{D_R} \oplus -\Delta^{N}_{\Omega_{-}}\bigr) \leq \inf \sigma_{ess} (-\Delta_{\sigma})\ .
\end{equation}
Since $\sigma_{ess}(-\Delta^{N}_{D_R})=\emptyset$ (note that $-\Delta^{N}_{D_R}$ is defined on a bounded Lipschitz domain and hence has only discrete spectrum, see Ref.~\citenum{Dob05}) we have
\begin{equation}\label{EquationINProof}
\inf \sigma_{ess}(-\Delta^{N}_{\Omega_{-}}) \leq \inf \sigma_{ess} (-\Delta_{\sigma})\ .
\end{equation}
Assume now that $\sigma(y)$ fulfils condition $(i)$: if we choose $R > L$ in the previous construction, $-\Delta^{N}_{\Omega_{-}}$ is a positive operator since $\sigma(y) \leq 0$ for a.e. $y \geq R$. Furthermore, using the same methods as in the proof of the first statement, we conclude that $[0,\infty) \subset \sigma_{ess}(-\Delta^{N}_{\Omega_{-}})$. Accordingly, by \eqref{EquationINProof} and by \eqref{EquationContainingEssentialSpectrum}, one concludes that $\sigma_{ess} (-\Delta_{\sigma})=[0,\infty)$.

Finally, assume that $\sigma(y)$ fulfils condition $(ii)$: choosing again $R > 0$ in the previous construction large enough such that $\|\sigma(y)\|_{L^{\infty}(\partial \rz^2_+ \setminus B_R(0))} < \varepsilon $ one has (in the same way as outlined for $-\Delta_{\sigma}$ in Remark~\ref{RemarkLowerBound})
\begin{equation}
\inf \sigma(-\Delta^{N}_{\Omega_{-}}) \geq -32 \varepsilon^2\ .
\end{equation}  
Since $\varepsilon > 0$ is arbitrary we conclude the statement while taking \eqref{EquationINProof} into account.
\end{proof}

\begin{rem}\label{RemarkConstantSigma}
Regarding Theorem~\ref{TheoremNeumann} it is instructive to consider a case where $\sigma(y)$ is (essentially) bounded and where the essential spectrum has indeed a negative part. For this we choose $\sigma(y)=\sigma > 0$ which means that the associated quadratic form is given by
\begin{equation}
q_{\sigma}[\varphi]:=\int_{\rz^2_+}|\nabla \varphi|^2\ \ud x-\sigma \int_{\partial \rz^2_+}|\varphi_{bv}(y)|^2 \ \ud y \ ,
\end{equation}
being defined on $\cD_q=H^1(\rz^2_+)$. The important observation is that the Laplacian $-\Delta_{\sigma}$ which is associated with $q_{\sigma}[\cdot]$ can be decomposed as
\begin{equation}\label{TensorDecomposition}
-\Delta_{\sigma}=-\Delta^{(1)}_{\sigma} \otimes \eins + \eins \otimes -\Delta^{(1)}_{\sigma}
\end{equation}
where $-\Delta^{(1)}_{\sigma}$ is the (self-adjoint) one-dimensional Laplacian $-\frac{\ud^2}{\ud x^2}$ defined on $\cD_{\sigma}=\{\varphi \in H^2(\rz_+)\ |\ \varphi^{\prime}(0)+\sigma\varphi(0)=0\}$. This follows from the fact that the quadratic form associated with $-\Delta^{(1)}_{\sigma}$ is 
\begin{equation}
q^{(1)}_{\sigma}[\varphi]:=\int_{\rz_+}|\varphi^{\prime}|^2\ \ud x-\sigma |\varphi(0)|^2\ ,
\end{equation}
being defined on $H^1(\rz_+)$~\cite{Kuc04}.
Also, it is well-known that $-\Delta^{(1)}_{\sigma}$ has exactly one eigenstate of negative energy at $-\sigma^2$ with multiplicity one (see, e.g., Ref.~\citenum{BL10}) and essential spectrum $\sigma_{ess}(-\Delta^{(1)}_{\sigma})=[0,\infty)$. Due to \eqref{TensorDecomposition} one can employ standard results (see, e.g., Refs.~\citenum{GustafsonSigal,schmudgen2012unbounded}) to conclude that $\sigma_d(-\Delta_{\sigma})=\{-2\sigma^2\}$, where $-2\sigma^2$ is an eigenvalue of multiplicity one, and $\sigma_{ess}(-\Delta_{\sigma})=[-\sigma^2,\infty)$. Furthermore, the (normalised) ground state associated with the eigenvalue $-2\sigma^2$ is given by $\varphi(x,y)=2\sigma e^{-\sigma(x+y)}$.
\end{rem}
\subsection{On the discrete spectrum}
In a second step we prove the existence of an eigenstate at negative energy whenever the two-particle interactions are \textit{attractive on average}, assuming in addition that $\sigma(y)$ is such that $\inf \sigma_{ess}(-\Delta_{\sigma}) =0$ (see condition $(i)$ or $(ii)$ as formulated in Theorem~\ref{TheoremNeumann}). Attractive on average shall mean $\sigma(y) \in L^1(\rz_+)$ and 
\begin{equation}\label{EventuallyAttractive}
\int_{0}^{\infty}\sigma(y)\ \ud y > 0 \ .
\end{equation}
%


From a physical point of view, condition \eqref{EventuallyAttractive} implies that the potential $\sigma(y)$ can exhibit some oscillatory behavior while $-\Delta_{\sigma}$ still maintains an eigenstate at negative energy. This, on the other hand, allows one to consider ``more natural'' interactions since repulsive interactions are always expected if the particles are too close to each other (hardcore repulsion).
\begin{theorem}\label{TheoremEssentialI} Assume that $\sigma(y) \in L^{\infty}(\rz_+)$ is such that $\inf \sigma_{ess}(-\Delta_{\sigma}) =0$ (compare with Theorem~\ref{TheoremNeumann}). If furthermore $\sigma(y) \in L^1(\rz_+)$ and $\int_{0}^{\infty}\sigma(y)\ud y > 0$ then there exists an eigenstate state below the essential spectrum.
\end{theorem}
\begin{proof} 
We use a test-function argument and consider the (unnormalised) function
\begin{equation}
\varphi_{\varepsilon}(r)=e^{-r^{\varepsilon}}\ , \quad \varepsilon >0\ ,
\end{equation}
using polar coordinates. We calculate 
\begin{equation}\label{EquationProofEigenstateIII}\begin{split}
q[\varphi_{\varepsilon}]&=\frac{\pi}{2}\int_{0}^{\infty}|\varphi^{\prime}_{\varepsilon}(r)|^2\ r \ \ud r-2\int_{0}^{\infty}\sigma(y) e^{-y^{\varepsilon}}\ \ud y \\
&=\frac{\pi \varepsilon^2 }{2}\int_{0}^{\infty}r^{2\varepsilon-1}e^{-2r^{\varepsilon}} \ \ud r-2\int_{0}^{\infty}\sigma(y) e^{-y^{\varepsilon}}\ \ud y \ .
\end{split}
\end{equation}
By dominated convergence and due to \eqref{EventuallyAttractive} the integral $2\int_{0}^{\infty}\sigma(y) e^{-y^{\varepsilon}}\ \ud y$ converges to some positive number in the limit $\varepsilon \rightarrow 0$. If we can show that the first integral converges to $0$, the statement follows since then $q[\varphi_{\varepsilon}]<0$ for some $\varepsilon$ small enough and $\sigma_{ess}(-\Delta_{\sigma})=[0,\infty)$ by Theorem~\ref{TheoremNeumann}.

Regarding the first integral we first employ an iterated integration by parts to yield
\begin{equation}\label{EquationProof1}
\int_{0}^{\infty}r^{2\varepsilon-1}e^{-2r^{\varepsilon}} \ \ud r =\frac{2^{n-2}}{(n-1)!}\int_{0}^{\infty}r^{n\varepsilon-1}e^{-2r^{\varepsilon}} \ \ud r 
\end{equation}
for any $n \geq 2$. Furthermore, employing a change of variables $(y=r^{\varepsilon})$ and setting $\varepsilon:=n^{-1}$ we obtain
\begin{equation}\label{EquationProof2}\begin{split}
\int_{0}^{\infty}e^{-2r^{\varepsilon}} \ \ud r&=n\int_{0}^{\infty}e^{-2y}y^{n-1} \ \ud y \\
&=n2^{-n}\Gamma(n)\ .
\end{split}
\end{equation}
Consequently, writing $\varepsilon:=n^{-1}$ and combining \eqref{EquationProof1} and \eqref{EquationProof2}, we can now calculate the first integral in \eqref{EquationProofEigenstateIII}: we obtain
\begin{equation}\begin{split}
\frac{\pi \varepsilon^2 }{2}\int_{0}^{\infty}r^{2\varepsilon-1}e^{-2r^{\varepsilon}} \ \ud r&=\frac{\pi}{2n^2}\frac{2^{n-2}}{(n-1)!}\int_{0}^{\infty}e^{-2r^{n^{-1}}} \ \ud r\\
&=\frac{\pi}{8n}\ .
\end{split}
\end{equation}
This shows that the first integral in \eqref{EquationProofEigenstateIII} converges to zero as $n \rightarrow \infty$, hence proving the statement.
\end{proof}
%
%
%
%
%
%
In a next step, assuming that $\supp \sigma(y) \subset [0,L]
$ for some $L > 0$, we estimate the number of eigenstates of $-\Delta_{\sigma}$ below the essential spectrum. As in the proof of Theorem~\ref{TheoremNeumann}, we use a comparison argument between $-\Delta_{\sigma}$ and now a direct sum of Laplacians $-\Delta_{\Omega_{1}} \oplus -\Delta_{\Omega_{2}}$ where $\rz^2_+=\Omega_{1}\ \dot{\cup}\ \Omega_{2}$ and where $\Omega_{1}$ is the triangle which is obtained by connecting the three points $(0,0)$,$(L,0)$ and $(0,L)$. On an operator level, we impose Neumann boundary conditions along the line connecting $(L,0)$ and $(0,L)$ as well as on $\partial \Omega_2$. Furthermore, Robin boundary conditions as in \eqref{RBConditions} with constant $\hat{\sigma}:=\|\sigma(y)\|_{\infty}$ are imposed along the other sides of the triangle $\Omega_1$.

On a form level, the quadratic form associated with $-\Delta_{\Omega_{1}}$ is given by
\begin{equation}
q_{3}[\varphi]:=\int_{\Omega_1} |\nabla \varphi|^2\ \ud x - \hat{\sigma} \int_{\partial \rz^2_+ \cap \partial \Omega_1}\ |\varphi_{bv}(y)|^2 \ \ud y\
\end{equation}
and it is defined on $H^1(\Omega_1)$. Furthermore, the quadratic form associated with $-\Delta_{\Omega_{2}}$ is given by
\begin{equation}
q_{4}[\varphi]:=\int_{\Omega_2} |\nabla \varphi|^2\ \ud x \ ,
\end{equation}
being defined on $H^1(\Omega_2)$. As a consequence, the operator $-\Delta_{\Omega_{1}} \oplus -\Delta_{\Omega_{2}}$ is (in the sense of an operator bracketing) smaller than $-\Delta_{\sigma}$. Hence, denoting by $N_{-}(A)$ the number of negative (discrete) eigenvalues of a self-adjoint operator $A$, we obtain the following result.
\begin{lemma}\label{UpperEstimate} Let $-\Delta_{\sigma}$ be given with $\supp \sigma(y) \subset [0,L]$ and corresponding $-\Delta_{\Omega_{1}}$ as constructed above. Then
\begin{equation}
N_{-}(-\Delta_{\sigma}) \leq N_{-}(-\Delta_{\Omega_{1}})\ .
\end{equation}
\end{lemma}
\begin{proof}
As mentioned already, in the sense of an operator bracketing we have $-\Delta_{\Omega_{1}} \oplus -\Delta_{\Omega_{2}} \leq -\Delta_{\sigma}$. A direct application of the minimax principle (see, e.g., Corollary~12.3 in Ref.~\citenum{schmudgen2012unbounded}) hence implies $N_{-}(-\Delta_{\sigma}) \leq N_{-}(-\Delta_{\Omega_{1}} \oplus -\Delta_{\Omega_{2}})$. The statement then follows by taking into account that $N_{-}(-\Delta_{\Omega_{1}} \oplus -\Delta_{\Omega_{2}})=N_{-}(-\Delta_{\Omega_{1}})$, noting that $-\Delta_{\Omega_{2}}$ is a positive operator.
\end{proof}

There is now an interesting way to estimate $N_{-}(-\Delta_{\Omega_{1}})$ further by reducing the two-particle (or two-dimensional) problem to a one-particle and hence one-dimensional problem. For this, let $-\tilde{\Delta}_{\hat{\sigma}}$ denote the one-dimensional Laplacian $-\frac{\ud^2}{\ud x^2}$ being defined on
\begin{equation}
\cD(-\tilde{\Delta}_{\sigma})=\{\varphi \in H^2(0,L)\ | \ \varphi^{\prime}(0)+\hat{\sigma} \varphi(0)=0\ \text{and} \ -\varphi^{\prime}(L)+\hat{\sigma} \varphi(L)=0 \ \}\ ,
\end{equation}
with $\hat{\sigma} > 0$ (this operator is a ``compact'' version of the operator considered in Remark~\ref{RemarkConstantSigma}). It is well-known that $\bigl(-\tilde{\Delta}_{\hat{\sigma}},\cD(-\tilde{\Delta}_{\hat{\sigma}})\bigr)$ is self-adjoint and has at most two eigenstates at negative energy (see, e.g., Ref.~\citenum{BL10}). Furthermore, the ground state eigenvalue $\varepsilon_0:=-\kappa^2$ corresponds to the solution of
\begin{equation}\label{ConditionGroundState}
\kappa \tanh{\Bigl(\frac{\kappa L}{2}\Bigr)}=\hat{\sigma} \ ,
\end{equation}
with $\kappa > \hat{\sigma} > 0$~\cite{BolEnd09}.

The next step is then to extend the triangle $\Omega_1$ on which the operator $-\Delta_{\Omega_{1}}$ is defined in a suitable way: i.e., $\Omega_1$ is extended by reflection across the line connecting $(L,0)$ and $(0,L)$ to a square $\tilde{\Omega}_1$ of side length $L$. By construction, Robin-boundary conditions with constant $\hat{\sigma} $ are imposed along the sides of $\tilde{\Omega}_1$. However, the Neumann boundary conditions which were present along line connecting $(L,0)$ and $(0,L)$ are now \textit{implicitly} implemented through requiring symmetry of the functions across this line. To repeat and to be more precise, we define the operator
\begin{equation}
-\tilde{\Delta}_{\hat{\sigma}}\otimes \eins + \eins \otimes -\tilde{\Delta}_{\hat{\sigma}}
\end{equation}
on the Hilbert space $L^2_s(\Omega)$ where $\Omega=[0,L]\times[0,L]$ and where the index $s$ refers to the fact that only functions which are symmetric across line connecting $(L,0)$ and $(0,L)$ are considered, i.e., $f(x,y)=f(L-y,L-x)$. By symmetry it is readily verified that the normal derivative of any (differentiable) function $f$ vanishes along the diagonal $y=L-x$.

If $\sigma(-\tilde{\Delta}_{\hat{\sigma}})=\{\varepsilon_n\ | \ n \in \nz_0\}$ denotes the spectrum of $-\tilde{\Delta}_{\hat{\sigma}}$ and $\{\varepsilon_n\}_{n \in \nz_0}$ are the corresponding eigenvalues in increasing order and counted with multiplicity, then
\begin{equation}
\sigma\bigl(-\tilde{\Delta}_{\hat{\sigma}}\otimes \eins + \eins \otimes -\tilde{\Delta}_{\hat{\sigma}}\bigr)=\{\varepsilon_n+\varepsilon_m\ | \ n,m \in \nz_0 \ \text{and}\ n \geq m \}\
\end{equation}
and $\{\varepsilon_n+\varepsilon_m\ | \ n,m \in \nz_0 \ \text{and}\ n \geq m \}$ are the eigenvalues of $-\tilde{\Delta}_{\hat{\sigma}}\otimes \eins + \eins \otimes -\tilde{\Delta}_{\hat{\sigma}}$, again counted with multiplicity: indeed, if $\{\varphi_n\}_{n \in \nz_0}$ are the eigenfunctions of $-\tilde{\Delta}_{\hat{\sigma}}$, each (tensor) product $\varphi_n \otimes \varphi_m$ forms an eigenfunction of $-\tilde{\Delta}_{\hat{\sigma}}\otimes \eins + \eins \otimes -\tilde{\Delta}_{\hat{\sigma}}$ to the eigenvalue $\varepsilon_n+\varepsilon_m$ if the underlying Hilbert space is $L^2(\Omega)$ (see, e.g., Ref.\citenum{GustafsonSigal}). However, since we are only working on the symmetric subspace $L^2_s(\Omega)$, we only have to consider symmetrised eigenfunctions, i.e., eigenfunctions of the form $\varphi_n \otimes \varphi_m+\varphi_m \otimes \varphi_n$. As a consequence, for each $n,m \in \nz_0$, the two pairs $\varphi_n \otimes \varphi_m$ and $\varphi_m \otimes \varphi_n$ yield only one symmetric eigenfunction and, in order to avoid an overcounting of eigenvalues, one therefore has to require $n \geq m$.

Since $\varepsilon_0=-\kappa^2$ we obtain, assuming that $\bigl(-\tilde{\Delta}_{\hat{\sigma}},\cD(-\tilde{\Delta}_{\hat{\sigma}})\bigr)$ has exactly one eigenstate at negative energy,
\begin{equation}
N_{-}(-\Delta_{\Omega_{1}}) \leq \#\{n \in \nz_0\ |\ \varepsilon_n < \kappa^2 \}\ .
\end{equation}
Furthermore, employing the results obtained in Ref.~\citenum{BL10}, we see that $\bigl(-\tilde{\Delta}_{\hat{\sigma}},\cD(-\tilde{\Delta}_{\hat{\sigma}})\bigr)$ has only one eigenstate at negative energy if and only if
\begin{equation}
\hat{\sigma} \leq \frac{2}{L}\ .
\end{equation}
Hence, summarising the conclusions from above we have established the subsequent statement.
\begin{theorem}\label{TheoremOneDimensionalEstimate} Let $-\Delta_{\sigma}$ be given with $\supp \sigma(y) \subset [0,L]$ and $\hat{\sigma}=\|\sigma(y)\|_{\infty} \leq \frac{2}{L}$. If $-\tilde{\Delta}_{\hat{\sigma}}$ denotes the one-dimensional Laplacian as constructed above with eigenvalues $\{\varepsilon_n\}_{n \in \nz_0}$, one has
\begin{equation}
N_{-}(-\Delta_{\sigma}) \leq \#\{n \in \nz_0\ |\ \varepsilon_n < |\varepsilon_0| \}\ .
\end{equation}
\end{theorem}
It is interesting to note that, using equation~(4.4) of Ref.~\citenum{BolEnd09}, all positive eigenvalues of \linebreak $\bigl(-\tilde{\Delta}_{\hat{\sigma}},\cD(-\tilde{\Delta}_{\hat{\sigma}})\bigr)$ can be written as $\varepsilon_n=k_n^2 > 0$ where $k_n$ is a (positive) solution of
\begin{equation}\label{QCondition}
\tan(kL)=\frac{2\hat{\sigma} k}{\hat{\sigma}^2-k^2}\ .
\end{equation}
As a consequence, combining equations \eqref{ConditionGroundState} and \eqref{QCondition} and using Theorem~\ref{TheoremOneDimensionalEstimate} we arrive at the following.
\begin{cor}\label{PropositionMost} For any fixed value $L > 0$ there exists $\sigma_0 > 0$ such that $-\Delta_{\sigma}$ has at most one eigenstate at negative energy for potentials $\sigma(y)$ with $\supp \sigma(y) \subset [0,L]$ and $\|\sigma(y)\|_{\infty} \leq \sigma_0$. If, in addition, $\sigma(y)$ fulfils \eqref{EventuallyAttractive} then there exists exactly one eigenstate.
\end{cor}
%
%
\subsection{On the ground state energy}
In this final subsection we provide upper bounds on the bottom of the spectrum. This result then provides us, in particular, with an estimate of the lowest eigenvalue below the essential spectrum in the case where such an eigenvalue exists.
\begin{theorem}\label{EstimateGroundStateEnergy} Assume $\sigma(y) \in L^{\infty}(\rz_+)$ and set $\hat{\sigma}:=\|\sigma(y)\|_{\infty}$. If $E_{\sigma}=\inf \sigma(-\Delta_{\sigma})$ denotes the ground state energy then
\begin{equation}
-2\hat{\sigma}^2\leq E_{\sigma} \leq -2\hat{\sigma}^2+8\hat{\sigma}^2\int_{0}^{\infty}[\hat{\sigma}-\sigma(y)]e^{-2\hat{\sigma}y}\ \ud y\ .
\end{equation}
\end{theorem}
\begin{proof} The lower bound follows from a comparison argument: For this consider the operator $-\Delta_{\hat{\sigma}}$ with $\sigma(y):=\hat{\sigma}$ as described in Remark~\ref{RemarkConstantSigma}. We readily see that this operator is indeed smaller (in the sense of an operator bracktering) than the given operator $-\Delta_{\sigma}$ with $\sigma(y)$. Furthermore, since the lowest spectral value of $-\Delta_{\hat{\sigma}}$ is the eigenvalue $-2\hat{\sigma}^2$, the lower bound follows directly.
	
To obtain the upper bound we use the (normalised) test-function $\varphi(x,y)=2\hat{\sigma}e^{-\hat{\sigma}(x+y)}$, i.e., the ground state of the example discussed in Remark~\ref{RemarkConstantSigma}. We calculate
\begin{equation}\begin{split}
q[\varphi]&=2\hat{\sigma}^2-8\hat{\sigma}^2\int_{0}^{\infty}\sigma(y)e^{-2\hat{\sigma}y}\ \ud y \\
&=-2\hat{\sigma}^2+8\hat{\sigma}^2\int_{0}^{\infty}[\hat{\sigma}-\sigma(y)]e^{-2\hat{\sigma}y}\ \ud y \ .
\end{split}
\end{equation}
This proves the claim.
\end{proof}
We can illustrate the result of Theorem~\ref{EstimateGroundStateEnergy} by the following example: Consider the step-potential
\begin{equation}
\sigma(y):=
\begin{cases}
\sigma & \text{if} \ \ 0 \leq y \leq L \\
   0   & \text{if} \ \ y > L
\end{cases}
\end{equation}
where $\sigma > 0$ is the depth of the potential and $L > 0$ the range of it. In this case we know that the essential spectrum starts at zero and that there exists a negative eigenvalue (see Theorem ~\ref{TheoremNeumann} and Theorem~\ref{TheoremEssentialI}). Consequently, Theorem~\ref{EstimateGroundStateEnergy} implies
\begin{equation}
-2\sigma^2\leq E_{\sigma} \leq -2\sigma^2+4\sigma^2 e^{-2\sigma L}\ ,
\end{equation}
leaving us with an estimate of the lowest eigenvalue of the system.
\section{Asymptotics of the ground state}
\label{Sec3}
In this final section we will discuss certain properties of the ground state $\varphi_0(x)$ (where $x \in \rz^2_+$) of the system. Note that $\varphi_0 \in H^1(\rz^2_+)$ with $\|\varphi_0\|_{L^2(\rz^2_+)}=1$ is called ground state of the system if 
\begin{equation}
q[\varphi_0]=E_{\sigma}=\inf \sigma(-\Delta_{\sigma}) > -\infty\ .
\end{equation}
In particular, we want to describe its asymptotic behavior as $|x| \to \infty$ in the case where the boundary potential $\sigma(y)$ has compact support. 

Adopting standard methods available from the theory of Schr\"{o}dinger operators and, in particular, the methods used in the proof of Theorem~11.8 in Ref.~\citenum{LiebLoss}, we directly arrive at the following result.
\begin{lemma} Assume that $\varphi_0(x)$ is a ground state of $-\Delta_{\sigma}$. Then $\varphi_0(x)$ is unique (up to a constant phase) and can be chosen to be real-valued and strictly positive in $\rz^2_+$. Furthermore, $\varphi_0(x)$ fulfils the eigenvalue equation $-\Delta_{\sigma} \varphi_0=E_{\sigma} \varphi_0$.
\end{lemma}
As a final result we derive (pointwise) upper bounds on $\varphi_0(x)$ that describe its asymptotic behavior as $|x| \to \infty$ assuming, for simplicity, that $\sigma(y)$ has compact support. The key ingredient is an approach based on the concept of (positive) supersolutions as described in Ref.~\citenum{Agmon}. More precisely, if $E_{\sigma} < 0$, the supersolution in our case will be the function
\begin{equation}
u(x)=\frac{e^{-\sqrt{|E_{\sigma}|}|x|}}{\sqrt{|x|}}
\end{equation}
considered on $\Omega_R:=\{x \in \rz^2 \ |\ |x| > R\}$ for some $R > 0$ large enough. It is called a supersolution since it fulfils $\bigl(-\Delta+|E_{\sigma}|\bigr)u(x)\geq 0$ on $\Omega_R$. In the proof we will also use the fact that the eigenfunction $\varphi_0$ can be extended (by reflection across the coordinate axes) to all of $\Omega_R$ and assumed to be continuous. Indeed, for $R$ large enough, the boundary conditions on $\partial \rz^2_+ \cap \Omega_R$ are Neumann boundary conditions only (here enters the assumption of compact support). Hence, standard regularity theory (see, e.g., Ref.~\citenum{GilTru83}) implies that $\varphi_0|_{\partial \rz^2_+ \cap \Omega_R}$ is in $H^2(\rz^2_+ \cap \Omega_R)$. After having extended $\varphi_0$ to $\Omega_R$ by reflection and then, using a cut-off function, to all of $\rz^2$, standard Sobolev embedding theorems imply continuity (see, e.g., Ref.~\citenum{SigalHislop}). Note that extending $\varphi_0$ to all of $\rz^2$ while remaining in $H^2(\rz^2 \cap \Omega_R)$ is possible due to the Neumann boundary conditions (vanishing normal derivative).
\begin{theorem}\label{TheoremAsymptoticsGroundState} Assume that $\supp \sigma(y) \subset [0,L]$ for some $L >0$. If $\varphi_0$ is the  ground state of $-\Delta_{\sigma}$ with ground state energy $E_{\sigma} < 0$ then, for each $R > L$, there exists a constant $c>0$ such that 
\begin{equation}
|\varphi_0(x)|\leq c\ \frac{e^{-\sqrt{|E_{\sigma}|}|x|}}{\sqrt{|x|}}\ , \qquad \forall x \in \rz^2_+: |x| > R+1\ .
\end{equation}
\end{theorem}
\begin{proof} We only sketch the proof. For more details see Theorem~2.7 and Theorem~3.2 in Ref.~\citenum{Agmon}. As described above, we will assume $\varphi_0(x)$ to be extended to and continuous on all of $\Omega_R$. Now, picking $R>L$ and defining $R_0:=R+1$ we choose $c > 0$ such that
\begin{equation}\label{EqProof1}
cu(x)-\varphi_0(x) > 0, \quad \forall x: |x|=R_0\ .
\end{equation}
Note that \eqref{EqProof1} is meaningful due to continuity of both functions $u(x)$ and $\varphi_0(x)$ on $\{x \ | \ |x|=R_0\}$. The key idea now is to prove that the function
\begin{equation}
u_0(x):=\bigl[\varphi_0(x)-cu(x)\bigr]_+
\end{equation}
vanishes on all of $\Omega_{R_0}$. Note that $[\cdot]_{+}$ denotes the positive part of a function.

From continuity we conclude that there exists $\delta > 0$ such that $u_0(x)=0$ for all $x$ s.t. $R_0 \leq |x| < R_0 + \delta$. Furthermore, we note that $\varphi_0(x)-cu(x)$ is a subsolution, i.e., $(-\Delta+|E_{\sigma}|)\bigl(\varphi_0(x)-cu(x)\bigr) \leq 0$ on $\Omega_{R_0}$. Lemma~2.9 of Ref.~\citenum{Agmon} then shows that $u_0(x)$ is also a non-negative subsolution (in the weak sense). Hence,
\begin{equation}
\int_{\Omega_{R_0}} \nabla u_0 \nabla (\zeta^2 u_0) \ \ud x+|E_{\sigma}|\int_{\Omega_{R_0}} \zeta^2 u^2_0 \ \ud x \leq 0\ ,
\end{equation}
for any real function $\zeta \in C^{\infty}_0(\Omega_{R_0})$. Using the identity (see equation~(2.25) in Ref.~\citenum{Agmon})
\begin{equation}
\nabla u_0 \nabla (\zeta^2 u_0)=|\nabla(\zeta u_0)|^2-u^2_0|\nabla \zeta|^2
\end{equation}
we obtain
\begin{equation}\label{AgmonProofEstimate}
\int_{\Omega_{R_0}}\Bigl(|\nabla(\zeta u_0)|^2+|E_{\sigma}|\zeta^2 u^2_0\Bigr) \ \ud x \leq \int_{\Omega_{R_0}} u^2_0|\nabla \zeta|^2 \ \ud x\ .
\end{equation}
Now, using \eqref{AgmonProofEstimate} and the identity (see equation~(2.6) in Ref.~\citenum{Agmon})
\begin{equation}
\nabla u \nabla \left(\frac{\zeta^2 u^2_0}{u}\right) = \left|\nabla (\zeta u_0)\right|^2-u^2\left|\nabla \left(\frac{\zeta u_0}{u}\right)\right|^2\ ,
\end{equation}
we obtain
\begin{equation}\label{EstimateTestFunctionShmuel}
\int_{\Omega_{R_0}} u^2\left|\nabla \Bigl(\frac{\zeta u_0}{u}\Bigr)\right|^2 \ \ud x \leq \int_{\Omega_{R_0}} u^2_0|\nabla \zeta|^2 \ \ud x\ ,
\end{equation}
in analogy to equation~(2.28) in Ref.~\citenum{Agmon}.

We then choose a suitable sequence $(\chi_n)_{n \in \nz}$ with $\chi_n(x):=\chi(\frac{x}{n})$, $\chi \in C^{\infty}_0(\rz^2)$, $0 \leq \chi \leq 1$ and $\chi(x)=1$ for $|x|\leq 1$. Employing the Lemma of Fatou as well as the estimate \eqref{EstimateTestFunctionShmuel} we obtain
\begin{equation}\begin{split}
\int_{\Omega_{R_0}} u^2\left|\nabla \Bigl(\frac{u_0}{u}\Bigr)\right|^2 \ \ud x &\leq \liminf_{n \to \infty}\int_{\Omega_{R_0}} u^2_0|\nabla \chi_n|^2 \ \ud x \\
&=0 \ .
\end{split}
\end{equation}
As a consequence, $u_0(x)=\lambda u(x)$ with some constant $\lambda$. Since $u > 0$ in $\Omega_{R_0}$ while $u_0$ vanishes in some strip around $R_0$ it follows that $\lambda=0$. Thus $u_0(x)=0$ and the theorem is proved.
\end{proof}
\begin{rem}\label{FinalRemark} Theorem~\ref{TheoremAsymptoticsGroundState} establishes the asymptotics of the ground state wave function $\varphi_0(x)$ as $|x| \rightarrow \infty$, assuming the boundary potential $\sigma(y)$ has compact support. In order to understand how the compactness of the support affects the asymptotics, it is instructive to look at the ground state of the example discussed in Remark~\ref{RemarkConstantSigma}, i.e., the case of $\sigma(y)=\sigma > 0$. In this case the ground state is given by $\varphi(x,y)=2\sigma e^{-\sigma(x+y)}$. Setting $x=y$ for simplicity, one has $|\varphi(x,x)|=2\sigma e^{-2\sigma x}=2\sigma e^{-\sqrt{|E_{\sigma}|} |x|}$, taking into account that $|x|=\sqrt{2}x$ and that $E_{\sigma}=-2\sigma^2$ is the corresponding ground state energy. We hence conclude that the finite range of the boundary potential leads, in the direction where $x=y$, to an additional decay factor of $1/\sqrt{|x|}$.
\end{rem}
%
%
\begin{acknowledgments}
JK would like to thank J.~Bolte and S.~Egger for helpful discussions. Furthermore, we would like to thank K.~Pankrashkin and the referee for valuable comments which led to an improvement of the paper.
\end{acknowledgments}
%
\def\cprime{$'$} \def\polhk#1{\setbox0=\hbox{#1}{\ooalign{\hidewidth
  \lower1.5ex\hbox{`}\hidewidth\crcr\unhbox0}}}

\end{document}